\documentclass[conference]{IEEEtran}
\IEEEoverridecommandlockouts
\usepackage{cite}
\usepackage{amssymb,amsfonts}
\usepackage{graphicx}
\usepackage{textcomp}
\usepackage{xcolor}
\usepackage{algorithm}
\usepackage{algorithmicx}
\usepackage{algpseudocode}
\usepackage{amsmath}
\usepackage{indentfirst}
\usepackage{multirow}
\usepackage{enumitem}
\usepackage{url}
\usepackage{hyperref}
\usepackage{amsthm} 
\newtheorem{theorem}{Theorem}
\newtheorem{definition}{Definition}

\def\BibTeX{{\rm B\kern-.05em{\sc i\kern-.025em b}\kern-.08em
    T\kern-.1667em\lower.7ex\hbox{E}\kern-.125emX}}
\begin{document}

\title{VulMCI : Code Splicing-based Pixel-row Oversampling for More Continuous Vulnerability Image Generation\\}

\author{
\IEEEauthorblockN{Tao Peng}
\IEEEauthorblockA{
\textit{School of Computer Science and Artificial Intelligence,} \\
\textit{Wuhan Textile University, Wuhan 430063, China} \\
\textit{Email: pt@wtu.edu.cn}
}
\and
\IEEEauthorblockN{Ling Gui}
\IEEEauthorblockA{
\textit{School of Computer Science and Artificial Intelligence,} \\
\textit{Wuhan Textile University, Wuhan 430063, China} \\
\textit{Email: 2315363137@mail.wtu.edu.cn}
}
\and
\IEEEauthorblockN{Yi Sun}
\IEEEauthorblockA{
\textit{School of Computer Science and Artificial Intelligence,} \\
\textit{Wuhan Textile University, Wuhan 430063, China} \\
\textit{Email: id.yisun@gmail.com}
}
\and
\IEEEauthorblockN{Lijun Cai}
\IEEEauthorblockA{
\textit{College of Computer Science and Electronic Engineering,} \\
\textit{Hunan University, Changsha 410082, China} \\
\textit{Email: ljcai@hnu.edu.cn}
}
\and
\IEEEauthorblockN{Rui Li}
\IEEEauthorblockA{
\textit{College of Software Engineering and Cyber Security,} \\
\textit{Dongguan University of Technology, Dongguan 523000, China} \\
\textit{Email: ruili@dgut.edu.cn}
}
\and
\IEEEauthorblockN{Qiang Zhu}
\IEEEauthorblockA{
\textit{School of Computer Science and Artificial Intelligence,} \\
\textit{Wuhan Textile University, Wuhan 430063, China} \\
\textit{Email: qzhu@wtu.edu.cn}
}
\and
\IEEEauthorblockN{Li Li}
\IEEEauthorblockA{
\textit{School of Computer Science and Artificial Intelligence,} \\
\textit{Wuhan Textile University, Wuhan 430063, China} \\
\textit{Email: lli@wtu.edu.cn}
}
\thanks{Corresponding authors: Yi Sun and Lijun Cai.}
\thanks{\url{https://github.com/guilingxz/VulMCI}}
}

\maketitle

\begin{abstract}
In recent years, the rapid development of deep learning technology has brought new prospects to the field of vulnerability detection. Many vulnerability detection methods transform source code into images for detection, but they neglect the issue of image quality generation. Since vulnerability images do not possess clear and continuous contours like object detection images, Convolutional Neural Networks (CNNs) tend to lose semantic information during the convolution and pooling process. Therefore, this paper proposes a pixel row oversampling method based on code line concatenation to generate more continuous code features, addressing the problem of discontinuous colors in code images. The feasibility of the proposed method is theoretically analyzed and verified. Based on these efforts, we introduce the vulnerability detection system VulMCI, which is tested on two datasets, SARD and NVD. Experimental results demonstrate that VulMCI outperforms eight state-of-the-art vulnerability detectors (i.e., Checkmarx, FlawFinder, RATS, VulDeePecker, SySeVR, Devign, VulCNN, and AMPLE) in accuracy on the SARD dataset. Compared to other image-based methods, VulMCI shows improvements in all metrics, including a 19.82\% increase in True Positive Rate (TPR), a 4.95\% increase in True Negative Rate (TNR), and an 8.89\% increase in accuracy (ACC). On the NVD real-world dataset, VulMCI achieves an average accuracy of 90.41\%.
\end{abstract}


\section{Introduction}

With the widespread use of computer networks and the ubiquity of the Internet, vulnerability detection has become increasingly crucial as malicious hackers and cybercriminals continuously seek novel methods to infiltrate systems and pilfer sensitive information. Furthermore, as software scales and complexities continue to grow, vulnerabilities within source code have become more insidious and pervasive, posing a severe threat to information security. According to the latest data from the National Vulnerability Database (NVD)\cite{nvd}, we can observe that in 2021, the number of disclosed vulnerabilities surpassed the twenty-thousand mark, and even though the figure decreased to slightly over thirteen thousand in 2022, it remains a substantial quantity. Consequently, the need for efficient and automated source code vulnerability detection has emerged as an urgent requirement to ensure the robustness and reliability of software systems.

Traditional vulnerability detection methods typically employ techniques  based on code similarity\cite{Jang_Woo_Brumley_2012, Kim_Woo_Lee_Oh_2017, Li_Ernst_2012, Li_Zou_Xu_Jin_Qi_Hu_2016, Pham_Nguyen_Nguyen_Nguyen_2010} or rule-based approaches\cite{FlawFinder_2021 , RoughAuditTool_2021 , Checkmarx_2021} to analyze source code in order to identify potential vulnerabilities. However, these  conventional methods exhibit several notable limitations. Firstly, they  often require a significant amount of manual work and rule formulation,  which becomes impractical when dealing with large-scale code  repositories. Secondly, they frequently struggle to address complex  vulnerability types, especially those with concealment and variants.  Finally, these methods tend to produce false positives and false  negatives when handling extensive code repositories, thus compromising  the accuracy and efficiency of the detection process.

Nevertheless, in recent years, the rapid advancement of deep learning  technology has brought new prospects to the field of vulnerability  detection. Deep learning methods have the ability to automatically learn patterns and features from extensive source code without the need for  manual intervention, rendering them highly scalable. Furthermore, deep  learning technology is effective in handling complex vulnerability  types, thereby enhancing the accuracy of vulnerability detection. These  advantages have positioned deep learning-based vulnerability detection  methods as a current focal point of research and attention.

Existing deep learning-based vulnerability detection systems utilize various methods, including processing code into textual representations using natural language processing\cite{Li_Zou_Xu_Ou_Jin_Wang_Deng_Zhong_2018 , Russell_Kim_Hamilton_Lazovich_Harer_Özdemir_Ellingwood_McConley_2018}, analyzing code structure graphs using graph neural networks\cite{Duan_Wu_Ji_Rui_Luo_Yang_Wu_2019 , Zhou_Liu_Siow_Du_Liu_2019}, and generating images from code for detection\cite{wu2022vulcnn}. While each of these methods has its merits, there is still room for improvement. In some prior studies\cite{Li_Zou_Xu_Ou_Jin_Wang_Deng_Zhong_2018 , Li_Zou_Xu_Jin_Zhu_Chen_2022}, custom Common Weakness Enumeration (CWE) vulnerability sets were used to match and locate vulnerabilities. However, this approach may lead to incomplete vulnerability coverage and instances where vulnerability labels do not match in real-world datasets. Additionally, many methods\cite{Li_Zou_Xu_Jin_Zhu_Chen_2022 , Lv_Peng_Chen_Liu_Hu_He_Jiang_Cao_2023} rely on the concept of code slicing to eliminate redundant information, depending not only on highly accurate vulnerability line localization but also risking the loss of complete program semantic information during slicing.

Furthermore, methods that transform source code into images, such as VulCNN\cite{wu2022vulcnn}, although preserving the complete semantics and structural information of functions, have been found through experiments to generate RGB images with significant instances of black and discontinuous scattered color points in the channels. This characteristic is considered detrimental to the classification task of Convolutional Neural Network (CNN) models. 

In this paper, we adopt the baseline method of converting source code into images and propose improvements to address the issue of discontinuity between rows of pixels in the images. We apply a pixel row oversampling algorithm to the source code lines with dependencies to enhance the continuity between code lines. The main advantage of this approach is to reduce data noise and minimize differences between adjacent data points, thereby making the model more robust to small variations in the input data. This contributes to improving model stability and generalization capability.

In summary, this paper makes the following contributions:

\begin{itemize}
    \item We present a novel framework for vulnerability image detection called VulMCI, which incorporates a method of pixel row oversampling using code control flow. This method generates code feature images with more continuous rows of pixels, effectively enhancing the identification of vulnerability features.

    \item We propose a pixel row oversampling algorithm based on code concatenation. To address the issue of discontinuous numerical values in directly generated code images, we utilize the relationships between nodes in the control flow graph to insert new code lines between adjacent rows. This balances local variations in the image and highlights important features, thereby improving the CNN's ability to extract key features and enhancing classification performance.
s
    \item We present a novel finding that when code is converted into image representation, there exist discontinuous and abrupt numerical changes. This discontinuity may negatively impact the training of Convolutional Neural Network (CNN) models, especially after the images undergo pooling operations, leading to the loss of significant row information.

\end{itemize}

\textbf{Paper organization.} The remaining sections of this paper are organized as follows. Section 2 presents the motivation of our paper. Section 3 introduces our method. Section 4 presents the experimental results. Section 5 discusses future work. Section 6 provides an overview of related work. Finally, Section 7 summarizes the key findings of this paper.

\section{Motivation}
Figure \ref{fig:Vul_img} depicts a vulnerability code image generated using our VulCNN\cite{wu2022vulcnn} method. We observed the presence of numerous redundant zero vectors and discontinuous color points, which are detrimental to the subsequent training of CNN models. This is attributed to the significant differences in feature representation between vulnerability detection tasks and object detection tasks when using images for classification. Object detection tasks typically require models to capture continuous edge features of target objects, which are retained to some extent after multiple convolution and pooling operations, as illustrated in Figure \ref{fig:pool}. However, in vulnerability detection tasks, the features of key statements often occupy a relatively small proportion of the data, and their distribution may be more dispersed. This leads to significant fluctuations in pixel values and the generation of discontinuous color points when directly converting vulnerability code into images. This scenario presents two main challenges for classical convolutional neural networks (CNNs) when processing vulnerability code images:
\begin{enumerate}
    \item \textbf{Information Loss}: Vulnerability code images are likely to lose crucial semantic information after convolution and pooling operations, especially concerning the features of key vulnerability statements. This poses difficulties in training and inference for vulnerability detection models.
    \item \textbf{Interference Noise}: In addition to key statements, vulnerability code images typically contain a large number of irrelevant statements, which may introduce additional interference noise and reduce the model's accurate understanding of vulnerability features.
\end{enumerate}
In response to these challenges, we propose a pixel-row oversampling method based on control flow graph-guided code splicing. By integrating the structured features of the code, we selectively splice correlated code lines to generate new sample rows. This approach aims to bridge code segments with flow relationships, enhancing contextual continuity and thereby improving the extraction of vulnerability features.
\begin{figure}
    \centering
    \includegraphics[width=1\linewidth]{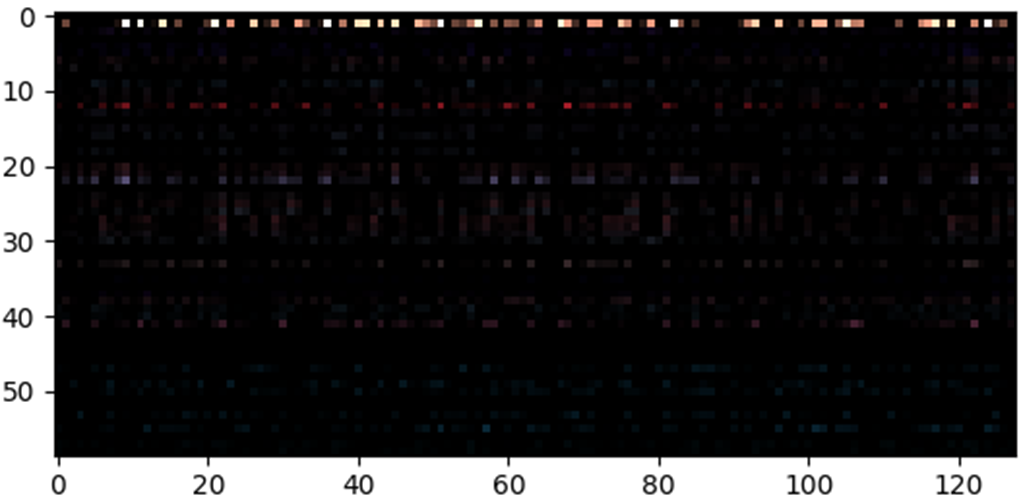}
    \caption{A vulnerability image generated using VulCNN method}
    \label{fig:Vul_img}
\end{figure}

\begin{figure}
    \centering
    \includegraphics[width=1\linewidth]{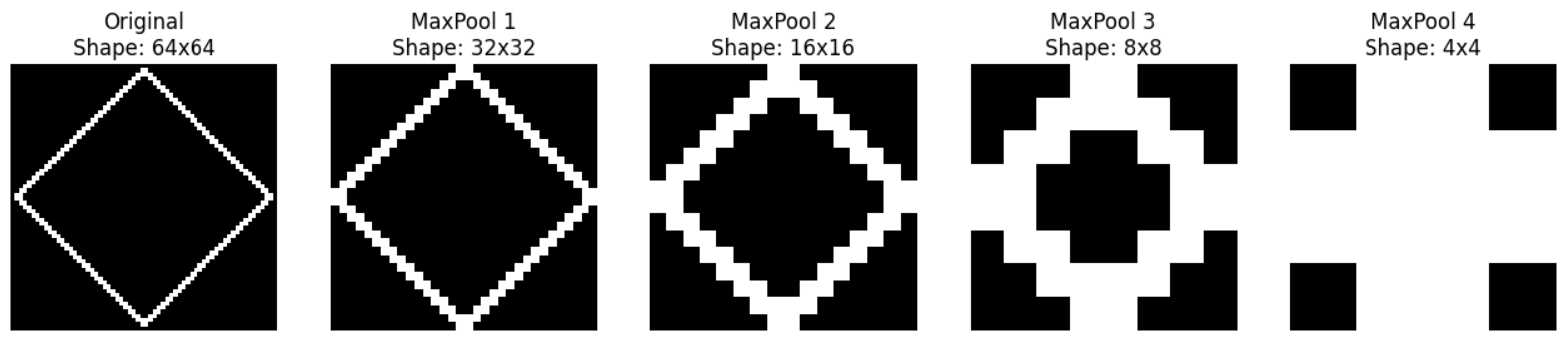}
    \caption{Comparison of Multiple Pooling Results}
    \label{fig:pool}
\end{figure}

\section{Our Method}
Existing methods for generating images from code overlook the continuity of pixel rows, resulting in low-quality images. Therefore, we propose the system approach of VulMCI to generate high-quality images with more continuous rows. As illustrated in Figure \ref{fig:cfg_system}, VulMCI consists of four main stages: constructing function samples, pixel row oversampling, code encoding and embedding, and generating grayscale images.
\begin{figure*}
    \centering
    \includegraphics[width=1\linewidth]{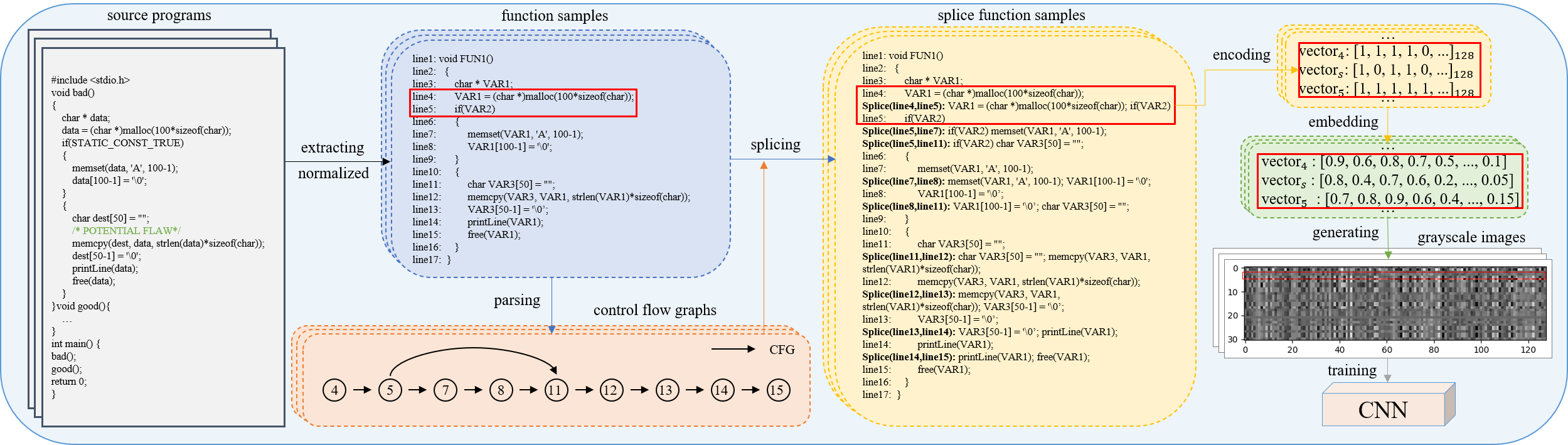}
    \caption{System overview of VulMCI}
    \label{fig:cfg_system}
\end{figure*}
\begin{itemize}
    \item \textbf{Building function samples:} Our system supports the extraction of function-level slice samples from initial program files, which undergo normalization. This process involves three main steps: removing comments, standardizing function names, and standardizing variable names.
    
    \item \textbf{Pixel-row oversampling:} We generate Code Property Graphs (CPGs) from the extracted function samples, and utilize the edge-node relationships of Control Flow Graphs (CFGs) to synthesize new code lines inserted into the source code, where each node represents a line of code.
    
    \item \textbf{Code encoding and embedding:} We utilize the sent2vec\cite{Pagliardini_Gupta_Jaggi_2018} model and parameters provided by VulCNN\cite{wu2022vulcnn} for code encoding and embedding.
    
    \item \textbf{Generating three-channel images:} We arrange the generated code line vectors together to form a grayscale image of vulnerabilities, which is used for subsequent training and classification by CNNs.
\end{itemize}

\subsection{Pixel-row oversampling}\label{AA}
We utilize Joern\cite{Yamaguchi_Golde_Arp_Rieck_2014} to generate Code Property Graphs (CPGs) from the extracted function samples, and then splice them based on the node relationships of Control Flow Graphs (CFGs). The control flow graph describes the control flow of the program, representing the transitions of control flow during program execution. Thus, splicing based on CFG node relationships enhances contextual connections. The splicing process follows these steps: we first traverse the line numbers of the code, and if a node corresponding to the line number has edges of CFG type, we concatenate the code of that node with the code of the target node completely, inserting it as a new code sample line after the node line. If the node corresponding to the code line does not belong to CFG, no splicing is performed, but the code line content is retained. As shown in the red-boxed portion of the function samples in Figure \ref{fig:cfg_system}, lines 4 and 5 correspond to nodes with existing edges in the CFG. Therefore, the codes of lines 4 and 5 are concatenated and inserted between the two lines. Finally, the function samples after pixel row oversampling are vectorized using sent2vec\cite{Pagliardini_Gupta_Jaggi_2018}, where each line of code is embedded as a row vector. These vector rows are then arranged together to generate a grayscale image. Figure\ref{fig:cfg_system} and Algorithm \ref{alg:alg1} describe how VulMCI converts function code into a grayscale image.

\begin{algorithm}
\caption{Converting the source code of a function into an grayscale image}
\label{alg:alg1}
\begin{algorithmic}[1]
\small
    \Require $F$: Source code of a function
    \Ensure $I$: An grayscale image
    \State $nF \gets \Call{CodeNormalization}{F}$
    \State $CFG \gets \Call{GraphExtraction}{nF}$
    \State $cfg\_nodes \gets \Call{ExtractCFGNodes}{CFG}$
    \State $\text{CFG\_adj\_matrix} \gets \Call{GenerateAdjacencyMatrix}{CFG}$
    \State $\text{channel} \gets []$
    \For{$i$ \textbf{in} $\text{range}(\text{len}(code\_list))$}
    \If{$i + 1$ \textbf{not in} $cfg\_nodes$}
        \State $\text{line\_vec} = \Call{EmbedSentence}{\text{code\_list}[i]}$
        \State $\text{channel.insert}(i, \text{line\_vec})$
    \EndIf
    \For{$j$ \textbf{in} $\text{range}(\text{len}(CFG\_adj\_matrix[i]))$}
        \If{$\text{CFG\_adj\_matrix}[i][j] = 1$}
            \State $\text{left\_code} = \text{code\_list}[i]$
            \State $\text{right\_code} = \text{code\_list}[j]$
            \State $\text{concatenated\_code} = \text{left\_code} + \text{right\_code}$
            \State $\text{line\_vec} = \Call{EmbedSentence}{\text{concatenated\_code}}$
            \State $\text{channel.append(line\_vec)}$
        \EndIf
    \EndFor
\EndFor
    \State $I = {Channel}$
\State \Return $I$
\end{algorithmic}
\end{algorithm}

\subsection{Theoretical Analysis}
\begin{definition}
For any two adjacent pixel rows represented by vectors $\text{vector1} = [s_1, s_2, ..., s_n]$ and $\text{vector2} = [f_1, f_2, ..., f_n]$, an image is considered pixel-row continuous when the absolute difference between corresponding elements, denoted as $|s_i - f_i| \leq \text{gap}$, where $\text{gap}$ is a small value. 

\end{definition}

\begin{theorem}
Our method has a good probability of producing more continuous images.
\end{theorem}
\begin{proof}  
\textbf{Step 1:} We demonstrate the continuity between adjacent pixel rows and concatenated pixel rows under specific conditions.
The objective of the sent2vec\cite{Pagliardini_Gupta_Jaggi_2018} model is to learn word and sentence embedding by minimizing a loss function. The simplified representation of the objective is as follows:
\begin{equation}
\min_{U,V} \sum_{S \in C} f_s(UV_{L_s})
\end{equation}
where $U$ represents the word embedding matrix, $V$ represents the sentence embedding matrix, $C$ is the corpus, and $L_s$ represents the word index list corresponding to the sentence $S$ in the corpus.The parameterized representation of the sentence vector is as follows. We aim to obtain a set of $UV$ parameters to demonstrate the effectiveness of our concatenated rows.
\begin{equation}
\text{vector1} = 
\begin{bmatrix} 
u_1\\
u_2\\
\vdots \\ 
\\ 
u_n\\ 
\end{bmatrix}
= UVL_u = \begin{bmatrix}
\sum_{i=1}^{m} u_{i1}w_{i1} \\
\sum_{i=1}^{m} u_{i2}w_{i2} \\
\vdots \\
\sum_{i=1}^{m} u_{in}w_{in} \\
\end{bmatrix}
\times L_u
\end{equation}

\begin{equation}
\text{concatenated\_vector}=
\begin{bmatrix} 
c_1\\
c_2\\
\vdots \\ 
\\ 
c_n\\ 
\end{bmatrix} 
= UVL_c = \begin{bmatrix}
\sum_{i=1}^{m} u_{i1}w_{i1} \\
\sum_{i=1}^{m} u_{i2}w_{i2} \\
\vdots \\
\sum_{i=1}^{m} u_{in}w_{in} \\
\end{bmatrix}
\times L_c
\end{equation}

\begin{equation}
\text{vector2}=
\begin{bmatrix} 
d_1\\
d_2\\
\vdots \\ 
\\ 
d_n\\ 
\end{bmatrix} 
= UVL_d = \begin{bmatrix}
\sum_{i=1}^{m} u_{i1}w_{i1} \\
\sum_{i=1}^{m} u_{i2}w_{i2} \\
\vdots \\
\sum_{i=1}^{m} u_{in}w_{in} \\
\end{bmatrix}
\times L_d
\end{equation}

If for any pair of elements, $|u_i - c_i| \leq \text{gap}$ and $|c_i - d_i| \leq \text{gap}$ holds, then the two vectors and the concatenated vector are considered continuous.\\
\textbf{Step 2:}We establish that the UV parameters in the training objective of the sent2vec model satisfy the requirements of the above inequality.
The complete training objective of the Sent2Vec\cite{Pagliardini_Gupta_Jaggi_2018} model, an unsupervised learning method for learning and inferring sentence embeddings, is as follows:
\begin{multline}
\min_{U,V} \sum_{S \in C} \sum_{w_t \in S} \Bigl( q_p(w_t) \ell \bigl( u_{w_t}^T v_{S_{w_t}} \bigr) \\
+ |N_{w_t}| \sum_{w' \in V} q_n(w') \ell \bigl( -u_{w'}^T v_{S_{w_t}} \bigr) \Bigr)
\end{multline}

where $U$ and $V$ represent the embedding matrices for words and sentences respectively, $C$ is the corpus containing the set of sentences to be learned. Additionally, $S$ represents a specific sentence in the corpus, and $w_t$ is the target word in sentence $S$. $u_{w_t}$ represents the embedding vector for the target word, and $v_{S_{w_t}}$ is the embedding vector for the sentence $S$ with the target word removed. $N_{w_t}$ represents the set of negative samples extracted for the target word $w_t$. The function $\ell(x)$ is the binary logistic regression loss function, measuring the model's predictive performance on positive and negative samples. By minimizing this loss function, the model learns embeddings that effectively capture semantic relationships between words and sentences in an unsupervised manner.

The first part of the loss function focuses on positive samples, encouraging similar sentences to be closer in the embedding space through the logistic regression loss function $\ell \left( u_{w_t}^T v_{S_{w_t}} \right)$. The second part addresses negative samples, pushing dissimilar sentences further apart in the embedding space through the logistic regression loss function $\ell \left( -u_{w'}^T v_{S_{w_t}} \right)$. The overall objective is to minimize the loss function by learning $U$ and $V$, resulting in embeddings with enhanced semantic relationships between sentences and words. By splicing and inserting adjacent code nodes into the middle based on the control flow graph, we introduce more contextual information, making contextually relevant code lines closer in the embedding space. This characteristic is reflected in the code encoding and embedding stage shown in Figure\ref{fig:cfg_system}. Therefore, the training objective of sent2vec\cite{Pagliardini_Gupta_Jaggi_2018} can obtain a set of UV parameters, theoretically leading to a smaller gap, proving the continuity between adjacent vector rows.
\end{proof}

\subsection{Classification}\label{AA}
In the realm of vulnerability detection, the adoption of advanced techniques becomes imperative to handle the complexity and nuances of source code. Convolutional Neural Networks (CNN) have emerged as a powerful tool in various image-based tasks, showcasing their ability to discern intricate features and patterns. In this section, we leverage the capabilities of CNN to address the challenges inherent in source code vulnerability detection.

After completing the image generation phase, we transformed the source code of functions into images. For a given image, we initiated training of a CNN model, subsequently applying this model for vulnerability detection. Since CNN models require input images of uniform size, and the number of code lines in different functions may vary, it becomes essential to select a suitable threshold length for extension or cropping.

The code length distribution of function samples in the experimental dataset (SARD) is illustrated in Figure~\ref{fig:code}. It is observed that the majority of function lengths are less than 100 lines. Therefore, we have chosen a code length threshold of 100 lines to generate our input images. And previous experiments have shown that choosing 100 lines as a threshold can balance overhead and detection performance well. For functions with code lengths less than 100, we pad the vectors with zeros at the end. In the case of functions exceeding 100 lines, we truncate the trailing portion of the vectors. The input image size is set to 100×128, 100 is the code line threshold, and 128 represents the dimensionality of the sentence vectors.
\begin{figure}
    \centering
    \includegraphics[width=1\linewidth]{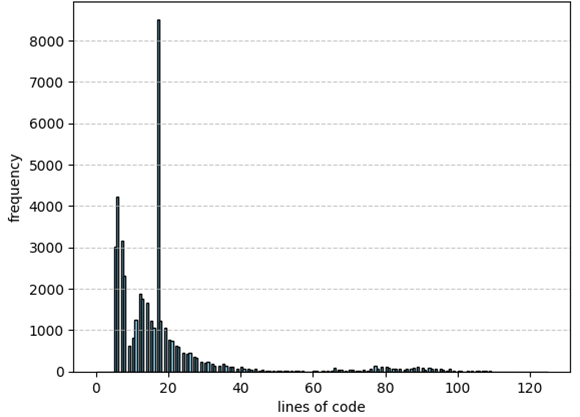}
    \caption{Function code length distribution}
    \label{fig:code}
\end{figure}

After generating images of fixed size, we build a Convolutional Neural Network (CNN) model through training on these images. As depicted in Figure~\ref{fig:CNN}, we employ convolutional filters of varying shapes, with dimensions m * 128, to ensure each filter can extract features across the entire space of the embedded sentences. The size of the filters roughly determines the length of the sentence sequences considered simultaneously. In VulMCI, we select 10 filters of different sizes (from 1 to 10), each containing 32 feature maps to capture features from different parts of the image. Following the max-pooling operation, the length of our fully connected layer is 320. Table~\ref{tab:Parameter} provides a detailed description of the parameters used in VulMCI. The entire model adopts Rectified Linear Unit (ReLU)\cite{Dahl_Sainath_Hinton_2013} as the non-linear activation function. Additionally, we utilize cross-entropy loss as the loss function in CNN for penalizing incorrect classifications. We employ Adam\cite{kingma2014adam} as the optimizer with a learning rate set to 0.001. Once the training is completed, we use the trained CNN model to classify new functions, determining whether they possess vulnerabilities.
\begin{figure}
    \centering
    \includegraphics[width=1\linewidth]{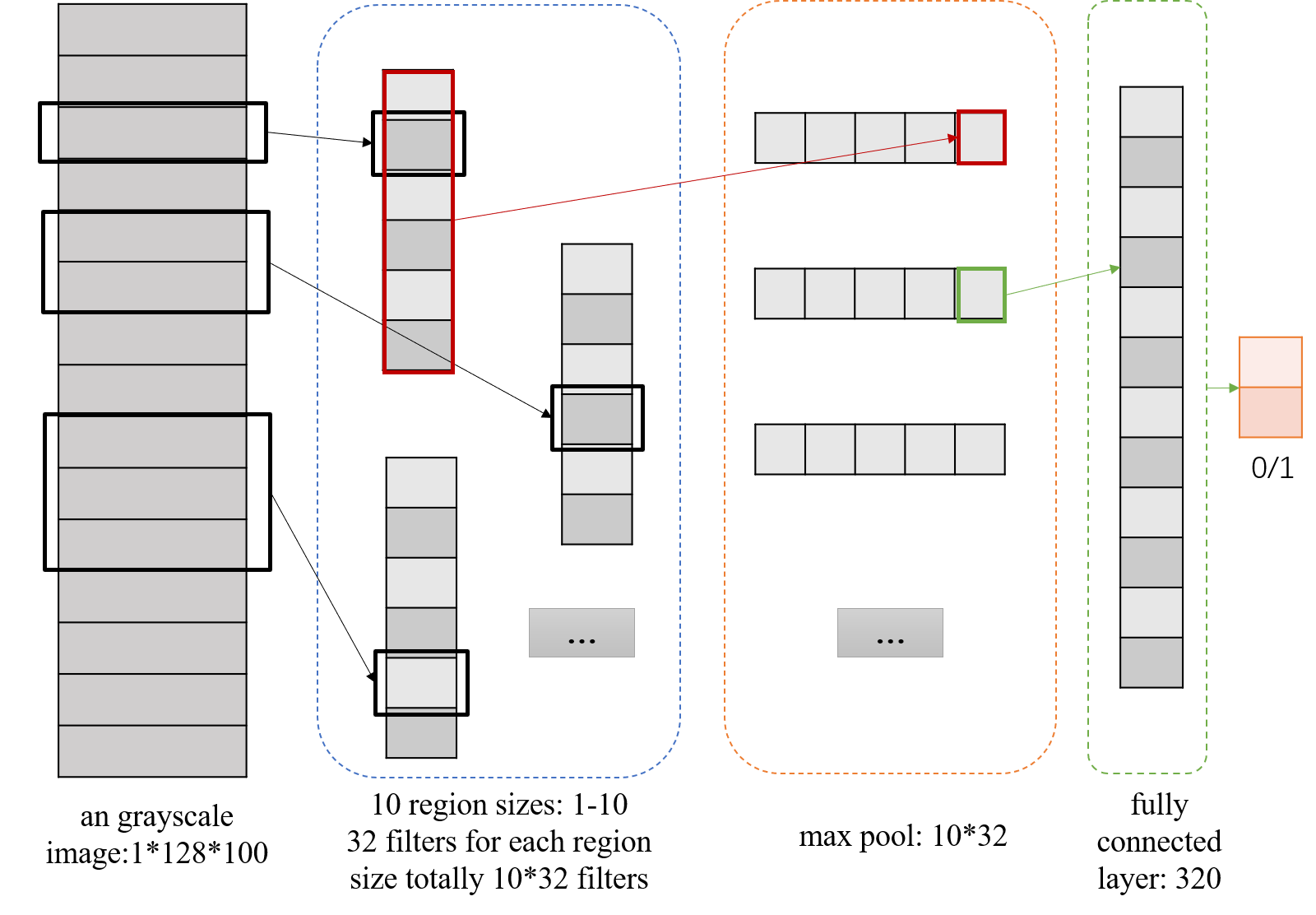}
    \caption{CNN classification of VulMCI}
    \label{fig:CNN}
\end{figure}

\begin{table}[htbp]
  \centering
  \caption{Parameter settings in VulMCI}
  \label{tab:Parameter}
  \begin{tabular}{c|c}
    \hline
    parameters & settings\\
    \hline
    loss function & Cross Entropy Loss \\
    activation function & ReLU\\
    optimizer & Adam\\
    batch size & 32\\
    learning rate & 0.001\\
    epoch num & 100\\
    \hline
  \end{tabular}
\end{table}

\section{EXPERIMENTS}

In this section, our aim is to address the following research questions:
\begin{itemize}
    \item \textbf{RQ1:} What are the advantages of VulMCI's pixel row oversampling method compared to other oversampling methods?
    
    \item \textbf{RQ2:} How does the detection performance of VulMCI compare to other state-of-the-art vulnerability detection systems?
    
    \item \textbf{RQ3:} Can VulMCI be applied to real-world vulnerability scanning?
\end{itemize}

\subsection{Experiment Settings}

The dataset used in this paper is sourced from the vulnerability dataset provided by SySeVR\cite{Li_Zou_Xu_Jin_Zhu_Chen_2022}. This dataset originates from SARD\cite{sard} and NVD\cite{nvd}, encompassing a total of 126 vulnerability types, each uniquely identified by a Common Weakness Enumeration (CWE) ID\cite{cwe}. The SARD dataset comprises production, synthetic, and academic programs (referred to as test cases) categorized as "good" (i.e., without vulnerabilities), "bad" (i.e., containing vulnerabilities), and "mixed" (i.e., vulnerabilities with available patch versions). The NVD dataset includes vulnerabilities in 19 popular C/C++ open-source products (software systems), along with possible diff files describing the differences between susceptible code and its patched versions. The SARD dataset contains a total of 21,233 program files, while the NVD dataset includes 2,011 real vulnerabilities and their corresponding fixed version files. Our system processed the dataset by extracting function samples from program files and removing function samples with less than ten lines of code. This exclusion was based on the observation that the majority of functions with fewer than ten lines only involve external function calls and cannot be easily classified as usable function samples. Ultimately, we obtained 12,116 vulnerability function samples and 4,660 non-vulnerability function samples from the SARD dataset, and 1,049 vulnerability function samples and 952 non-vulnerability function samples from the NVD dataset.

In this study, we first conducted comparative and ablation experiments using the SARD dataset, and then employed the NVD dataset in RQ3 to evaluate the effectiveness of vulnerability detection on real vulnerabilities. We adopted a k-fold cross-validation (k=5) approach to partition the entire dataset into 5 mutually exclusive subsets, with 4 subsets used for model training and 1 subset for model testing. This partitioning approach ensures the utilization of distinct data subsets for both training and testing phases, facilitating the assessment of the model's generalization performance. The devices and experimental parameters employed in this study are presented in Tables~\ref{tab:Parameter} and \ref{tab:devices}, respectively. To comprehensively evaluate the performance of our proposed method, we utilized multiple evaluation metrics, including False Positive Rate (FPR), False Negative Rate (FNR), Precision (Pr), Recall (Re), F1 Score (F1), and Accuracy (ACC). Each metric provides detailed insights into the algorithm's performance across different aspects, offering a nuanced understanding of its strengths and limitations.

\begin{table}[htbp]
\centering
  \caption{Details of the experimental devices}
  \label{tab:devices}
\begin{tabular}{lll}
\hline
Device                   & Type                 & Version  \\ \hline
GPU                      & Tesla                & V100     \\
CPU                      & Intel Xeon Silver    & 6130     \\
Operating system         & Centos linux release & 7.9.2009 \\ \hline
\multirow{6}{*}{Package} & Python               & 3.8      \\
                         & Torch                & 1.12.1   \\
                         & Matplotlib           & 3.8.2    \\
                         & Numpy                & 1.25.2   \\
                         & Networkx             & 3.1      \\
                         & Joern                & 2.0.121  \\ \hline
\end{tabular}
\end{table}

\subsection{Experimental Results}
\subsubsection{Experiments for Answering RQ1}
To validate the effectiveness of the CFG-based pixel row oversampling method in VulMCI, we additionally designed several oversampling methods for comparison. The specific methods are as follows:
\begin{enumerate}
\item To determine the influence of oversampling iterations on the results, we designed a method based on the K-nearest neighbor (KNN) adjacent line code splicing according to the semantic properties predicted by sent2vec. This method adopts segmented combination splicing to generate new lines of code. Given two adjacent lines of code, Line $i$ and Line $i+1$ (where $i$ represents the line number), adaptive splicing is performed between the two adjacent lines of code. The splicing process is executed according to the following steps:
\begin{itemize}
    \item From right to left in Line $i$, we take the number obtained from Equation 1 as the length of the substring, and then obtain the substring $left\_tokens$.
    
    \[    \text{num\_left\_tokens}[i] = \left\lfloor \frac{\text{left\_length} \cdot (k - i)}{k} \right\rfloor
    \]
    
    \item From left to right in Line $i+1$, we take the number obtained from Equation 2 as the length of the substring, and then obtain the substring $right\_tokens$.
    
    \[    \text{num\_right\_tokens}[i] = \left\lfloor \frac{\text{right\_length} \cdot i}{k} \right\rfloor
    \]
    
    \item The new concatenated line is obtained as $concatenated\_code = left\_tokens + right\_tokens$, and we complete the words at the character level.
    
    \item Where $i \in [1, k)$, k-1 new concatenated lines will be gradually inserted between adjacent lines, where k=1 indicates no splicing, k=2 indicates splicing once, k=3 indicates splicing twice, and so on.\end{itemize}

\begin{table*}[h]
\centering
\caption{Parameter analysis of k-valued adjacent code concatenation method}
\label{tab:k}
\begin{tabular}{cccccccc}
\hline
k & Remark & FPR & FNR & Pr & Re & F1 & ACC \\ 
\hline
1 & Baseline without splicing & 17.976 & 8.932 & 92.225 & 93.412 & 91.068 & 88.687 \\
2 & Splicing and inserting 1 line & 16.429 & 8.762 & 92.577 & 93.955 & 91.238 & 89.22 \\
3 & Splicing and inserting 2 lines & 18.69 & 8.252 & 92.476 & 93.215 & 91.748 & 89 \\
4 & Splicing and inserting 3 lines & 19.762 & 8.124 & 92.367 & 92.863 & 91.876 & 88.812 \\
5 & Splicing and inserting 4 lines & 17.738 & 8.72 & 92.381 & 93.508 & 91.28 & 88.906 \\
6 & Splicing and inserting 5 lines & 21.19 & 7.359 & 92.543 & 92.445 & 92.641 & 89 \\
7 & Splicing and inserting 6 lines & 20.476 & 7.954 & 92.341 & 92.637 & 92.046 & 88.75 \\
\hline
VulMCI\_CFG\_k=2 & Splicing only CFG nodes, using k=2 splicing strategy & 1.931 & 3.383 & 97.91 & 99.237 & 96.617 & 97.02 \\ 
\hline
\end{tabular}
\end{table*}
The experimental results, as depicted in Table\ref{tab:k}, encompassed a total of 7 scenarios. From these experiments, it was observed that the best performance was achieved when $k = 2$, indicating that inserting a single line yielded the optimal outcome. This approach resulted in a 0.533\% increase in accuracy compared to not splicing lines. However, the accuracy was still 7.8\% lower compared to VulMCI's CFG-based method. This discrepancy can be attributed to the indiscriminate oversampling of all lines using the K-nearest neighbor splicing method. While this approach enhanced the continuity of the images, it also fused redundant code segments, including those weakly associated with vulnerabilities. As vulnerability-related code segments are in the minority, their significance was not significantly bolstered by the splicing process.

\item To assess the impact of different oversampling methods on the results, we devised three additional control methods, including the complete concatenation of adjacent code and an oversampling approach where code lines are first vectorized and then inserted by adding their respective vectors. The complete concatenation of adjacent code involves inserting a single line between adjacent lines, where concatenated\_code = left\_code + right\_code. The oversampling method of inserting vectors involves the insertion vector being equal to the sum of the left and right vectors. To prevent excessively large values resulting from vector addition, we also experimented with dividing the sum of the vector rows by two before insertion, denoted as Insertion\_vector = (left\_vector + right\_vector) / 2.

\begin{table*}[h]
\centering
\caption{Comparison of other splicing methods}
\label{tab:all}
\begin{tabular}{cccccccc}
\hline
Method            & Description              & FPR & FNR & F1 & Pr & Re & ACC \\ \hline
No splicing       & Baseline                                  & 17.976 & 8.932 & 92.225 & 93.412    & 91.068 & 88.687   \\
All               & Complete code concatenation of adjacent lines & 17.5   & 8.464 & 92.559 & 93.606    & 91.536 & 89.157   \\
Vec               & Vectorize first, then add adjacent vector rows & 19.524 & 8.379 & 92.268 & 92.925    & 91.621 & 88.687   \\
Vec2              & Divide the sum of adjacent vector rows by two before insertion & 22.857 & 7.529 & 92.177 & 91.885    & 92.471 & 88.436   \\ \hline
VulMCI\_CFG\_All & Only splice CFG nodes and adopt the splicing strategy "All" & 1.609  & 3.259 & 98.035 & 99.364    & 96.741 & 97.199   \\ \hline
\end{tabular}
\end{table*}

The experimental results, as shown in Table\ref{tab:all}, indicate that the methods directly processing vector rows (vec and vec2) perform poorly, with vec2 even performing slightly worse than no splicing, suggesting that directly oversampling vector rows may lead to information loss or confusion. The methods of complete code concatenation and adaptive concatenation of 1 line (k=2) yield similar results, hence we compared them with the CFG-based method of adjacent lines with complete concatenation (VulMCI\_CFG\_all), which shows slightly better performance than VulMCI\_CFG\_k=2. This difference in results may be related to the distribution of the concatenated data.

\item As the Code Property Graph (CPG) contains various types of relationship subgraphs, to investigate which type of edge concatenation performs best in experiments, we compared three types of concatenation: Data Dependency Graph (DDG), Control Dependency Graph (CDG), and Control Flow Graph (CFG). Since in the above experiments, the concatenation methods 'k=2' and 'all' performed the best and yielded similar results, we also applied two different concatenation methods to the three types of edges.

\begin{table*}[]
\centering
\caption{Comparison of other splicing methods}
\label{tab:cpg}
\begin{tabular}{cccccccc}
\hline
Method & Description & FPR & FNR & F1 & Pr & Re & ACC \\ 
\hline
No Concatenation & Baseline & 17.976 & 8.932 & 92.225 & 93.412 & 91.068 & 88.687 \\
Concatenate All & Concatenate DDG, CFG, and CDG, complete code & 17.619 & 8.55 & 92.493 & 93.56 & 91.45 & 89.063 \\ 
\hline
DDG\_k=2 & Concatenate only DDG nodes, split-combine approach & 3.004 & 3.053 & 97.876 & 98.823 & 96.947 & 96.961 \\
CDG\_k=2 & Concatenate only CDG nodes, split-combine approach & 15.88 & 11.056 & 91.201 & 93.576 & 88.944 & 87.604 \\
CFG\_k=2 & Concatenate only CFG nodes, split-combine approach & 1.931 & 3.383 & 97.91 & 99.237 & 96.617 & 97.02 \\ 
\hline
DDG\_all & Concatenate only DDG nodes, Complete splice approach & 2.575 & 3.012 & 97.979 & 98.989 & 96.988 & 97.11 \\
CDG\_all & Concatenate only CDG nodes, Complete splice approach & 12.768 & 10.52 & 92.063 & 94.799 & 89.48 & 88.856 \\
CFG\_all & Concatenate only CFG nodes, Complete splice approach & 1.609 & 3.259 & 98.035 & 99.364 & 96.741 & 97.199 \\ 
\hline
\end{tabular}
\end{table*}

The experimental results, as shown in Table 5, indicate that when concatenating based on edge relationships, concatenating complete adjacent lines yields better results than using split-combine concatenation. This is related to the data distribution. When all lines of code are concatenated indiscriminately, the distances between lines are closer, making split-combine concatenation more suitable for this data distribution. The sent2vec model can predict the concatenated lines well by connecting contextual semantics. However, when concatenating based on code structural relationships, the concatenated code is more directional and strongly correlated with the vulnerability execution process, resulting in cases of concatenation across lines. Therefore, concatenating complete adjacent lines based on code structural information is more suitable. Among the three edge types, concatenating based on CFG yields the best results, with an F1 score of 98.035\% and an accuracy of 97.199\%. This is because CFG can better capture code structure and control flow. Specifically, code execution follows temporal logic, and concatenating based on the dependencies of data flow and code control flow does not disrupt the temporal relationship of the code. It also narrows the distance between key statements, facilitating vulnerability feature extraction. Concatenating based on DDG yields the second-best results because DDG only focuses on vulnerabilities caused by data, such as buffer overflow, and does not well reflect the formation relationships of logical vulnerabilities. Concatenating based on CDG yields the worst results, although slightly better than no concatenation. This is because CDG generally only contains conditional branches, such as if and while. Finally, when concatenating all three edge types of DDG, CFG, and CDG, the accuracy is only 0.376\% higher than no concatenation. It is evident that simultaneously concatenating three types of edges leads to many repeated concatenations between nodes, resulting in a loss of distinction between key vulnerability statements, thus having a counterproductive effect.

\end{enumerate}

\subsubsection{Experiments for Answering RQ2}

In this section, we compare VulMCI with several vulnerability detection tools, including a commercial static vulnerability detection tool (i.e., Checkmarx\cite{Checkmarx_2021}), two open-source static analysis tools (i.e., FlawFinde\cite{FlawFinder_2021} and RATS\cite{RoughAuditTool_2021}), and five deep learning-based vulnerability detection methods (i.e., VulDeePecker\cite{Li_Zou_Xu_Ou_Jin_Wang_Deng_Zhong_2018}, SySeVR\cite{Li_Zou_Xu_Jin_Zhu_Chen_2022}, Devign\cite{Zhou_Liu_Siow_Du_Liu_2019}, VulCNN\cite{wu2022vulcnn} and AMPLE\cite{Wen_Chen_Gao_Zhang_M.Zhang_Liao_2023}). Figure~\ref{fig:compare} presents a bar chart that intuitively illustrates the metrics of True Positive Rate (TPR), True Negative Rate (TNR), and Accuracy (ACC) for each of these methods.

\begin{figure}
    \centering
    \includegraphics[width=1\linewidth]{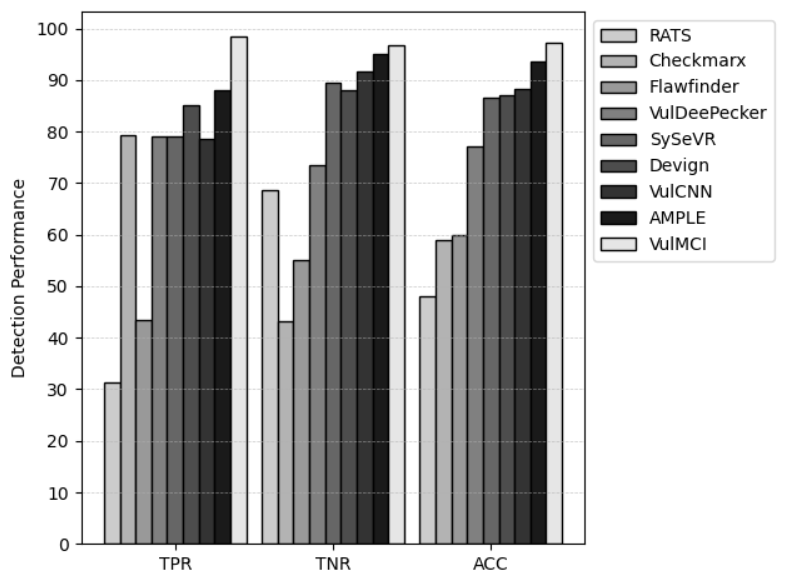}
    \caption{True Positive Rate (TPR), True Negative Rate (TNR), and Accuracy of  RATS, Checkmarx, FlawFinder, VulDeePecker, SySeVR,  Devign, VulCNN and AMPLE on detecting vulnerability}
    \label{fig:compare}
\end{figure}

For the three static analysis tools, RATS, Checkmarx, and FlawFinder, their detection effectiveness is not satisfactory, with accuracy rates all below 60\%. One plausible explanation is that they rely on human experts to define vulnerability rules for detection. However, as vulnerability types become increasingly complex, human experts cannot comprehensively define patterns for all vulnerabilities, leading to higher rates of false positives and false negatives.

Regarding VulDeePecker and SySeVR, both methods employ program slicing to process datasets, vectorizing slices for training bidirectional recurrent neural networks (BRNN) for vulnerability detection. The performance disparity between these two systems arises from VulDeePecker only considering vulnerabilities caused by API function calls and solely focusing on data dependency information. In contrast, SySeVR considers semantic information resulting from both data and control dependencies and supports a wider range of vulnerability syntax labels for matching vulnerabilities. Consequently, SySeVR achieves a 9.47\% higher accuracy than VulDeePecker.

SySeVR utilizes code slicing for vulnerability feature extraction, leading to incomplete code semantics. On the other hand, VulMCI enhances the semantic information of critical code while retaining complete function code. Additionally, the limited coverage of vulnerability label sets in slicing methods means that a function sample may have multiple vulnerability candidates, making precise vulnerability localization and slicing challenging. VulMCI achieves a 10.70\% higher accuracy than SySeVR.

Devign employs a graph neural network approach, effectively leveraging node relationships within the graph structure. However, its use of compound graphs with intricate network node information introduces redundant information weakly correlated with vulnerability nodes, increasing model complexity and performance overhead. Moreover, as the number of nodes increases, the model struggles to learn edge node information, leading to information loss. In contrast, VulMCI utilizes an image-based approach for detection, enhancing vulnerability features through oversampling based on code structure while retaining complete code semantics, thus avoiding the aforementioned issues. 

AMPLE addresses Devign's shortcomings by simplifying nodes according to type and variable, representing different types of edges with weighted vectors, and enhancing node representation using a multi-head attention mechanism. This results in a 6.41\% increase in accuracy compared to Devign. However, due to the addition of modules, AMPLE still incurs considerable performance overhead. In contrast, VulMCI achieves a higher simplification rate by mapping the complex node network to a network between line numbers in the preprocessing phase. Experimental results demonstrate that VulMCI outperforms AMPLE by 3.60\% in accuracy.

VulCNN utilizes centrality analysis to transform time-consuming graph analysis into efficient image scanning, effective for large-scale scanning. However, it overlooks the discontinuity of pixel information in images, leading to insufficient vulnerability feature extraction. VulMCI addresses this issue through pixel row oversampling, resulting in a 19.82\% increase in true positive rate (TPR), a 4.95\% increase in true negative rate (TNR), and an 8.89\% increase in accuracy (ACC).

\subsubsection{Experiments for Answering RQ3}
To evaluate VulMCI's real-world vulnerability detection capabilities, we trained and tested the model using 2011 real vulnerabilities and their corresponding fixed version files from the NVD. The data partitioning method remained consistent with previous experiments. Due to the scarcity of real-world vulnerability data and to prevent inaccuracies caused by individual sample variations, we conducted five repeated experiments, saving the results of the final generation. Each fold of the test set was evaluated in each experiment, and the average result of each fold's testing was computed. The average F1 score over the five repeated experiments reached 89.07\%, with an average accuracy of 90.41\%.

Real-world vulnerability detection often involves large-scale code scanning, necessitating consideration of performance overhead. The time overhead of each system is depicted in Figure 8, with VulCNN exhibiting excellent performance in this regard. We augmented VulCNN with a pixel row oversampling module to reduce the computation of the three central indicators for the PDG and trained it using grayscale images. Meanwhile, VulMCI also selected 100 lines as the function length threshold without increasing model training overhead. Figure \ref{fig:time} illustrates the time overhead proportions of the six models, among which VulMCI's time overhead is similar to VulCNN's, approximately one-sixth of Devign's, and approximately one-fourth of SySeVR's.

\begin{figure}
    \centering
    \includegraphics[width=1\linewidth]{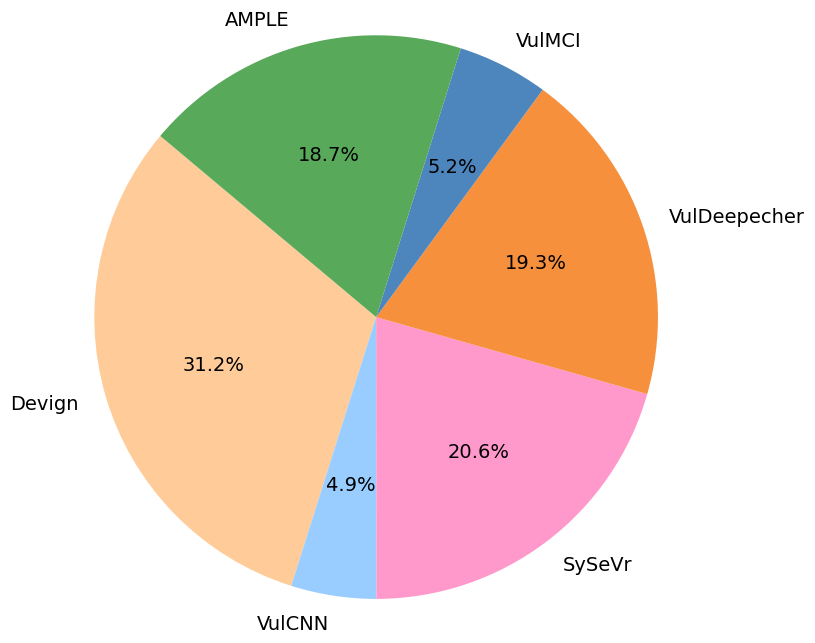}
    \caption{Time Overhead Scale Plot for Models}
    \label{fig:time}
\end{figure}

\section{Discussion}

Throughout the entire system workflow, the static analysis for generating the Code Property Graph (CPG) incurs the highest computational cost. Processing a dataset comprising twenty thousand function samples requires over twenty hours. Additionally, there is some redundancy in the label information of analyzed nodes. We are currently contemplating the reconstruction of a new, lightweight code analysis tool to enhance operational efficiency.

While utilizing the pixel row oversampling method to generate more continuous code line features, we have focused solely on oversampling between code lines. In other words, we have considered only the continuity of pixel rows above and below the code lines, neglecting the continuity within pixel rows. However, we are committed to ongoing research to substantiate the effectiveness of this concept.

In subsequent phases, we plan to extend our method to other graph-based detection systems to further validate its scalability. This extension aims to assess the adaptability and effectiveness of our approach beyond the current system.

\section{RELATED WORK}

In recent years, researchers have proposed various methods for vulnerability detection, categorizing static detection techniques into code similarity-based and pattern-based methods. Code similarity-based approaches, such as \cite{Kim_Woo_Lee_Oh_2017 , Li_Ernst_2012 , Pham_Nguyen_Nguyen_Nguyen_2010 , Sajnani_Saini_Svajlenko_Roy_Lopes_2016}, are generally suitable for identifying code clones or repetitions but exhibit poor performance in detecting unknown vulnerabilities.

Pattern-based methods encompass several directions. Rule-based methods rely on human experts to manually construct vulnerability features. Tools like Checkmarx\cite{Checkmarx_2021}, FlawFinder\cite{FlawFinder_2021}, and RATS\cite{RoughAuditTool_2021} using this approach often struggle to ensure coverage in practical detection, requiring a high level of expertise and thus yielding suboptimal results. Early machine-learning-based vulnerability detection methods, such as those based on metric program representations, typically employed feature engineering to manually extract vulnerability-related features. Classic methods include code churn\cite{perl2015vccfinder,zimmermann2010searching,morrison2015challenges,meneely2013patch,shin2011evaluating,shin2013traditional,walden2014predicting}, code complexity\cite{zimmermann2010searching,moshtari2013using,morrison2015challenges,younis2016fear,shin2011evaluating,shin2013traditional,walden2014predicting}, coverage\cite{zimmermann2010searching,morrison2015challenges}, dependency\cite{zimmermann2010searching,morrison2015challenges}, organizational\cite{zimmermann2010searching,morrison2015challenges}, and developer activity\cite{perl2015vccfinder,bosu2014identifying,meneely2013patch,shin2011evaluating,meneely2010strengthening}. In recent years, researchers have started using deep learning methods to automatically extract vulnerability features.

Due to different program representations at various compilation stages, deep learning-based program representation methods can be categorized into three types based on the organizational form of code representation: sequence-based\cite{grieco2016toward,wu2017vulnerability,russell2018automated,yan2021han,li2021vulnerability,tian2020bvdetector,li2021vuldeelocator}, syntax tree-based\cite{li2019improving,tanwar2021multi,zhang2021isvsf,wang2016automatically,lin2017poster}, and graph-based\cite{duan2019vulsniper,duan2020vulnerability,cao2020ftclnet,li2021acgvd,cheng2021deepwukong,zheng2021vulspg,wang2020combining,okun2013report}. VulDeePecker\cite{Li_Zou_Xu_Ou_Jin_Wang_Deng_Zhong_2018} slices the program, collects code snippets, and transforms them into corresponding vector representations. Ultimately, these vectors are used to train a Bidirectional Long Short-Term Memory (BLSTM) model for vulnerability detection. SySeVR\cite{Li_Zou_Xu_Jin_Zhu_Chen_2022} first performs syntactic analysis and then semantic slicing on vulnerable code, transforming sliced code into fixed-length vectors fed into a Bidirectional Gated Recurrent Unit (BGRU) model for training. VulDeeLocator\cite{li2021vuldeelocator} utilizes intermediate code to define program slices for vulnerability detection, proposing a new variant of Bidirectional Recurrent Neural Network (BRNN), namely BRNN-vdl, for vulnerability detection and localization. Devign\cite{Zhou_Liu_Siow_Du_Liu_2019} is a graph neural network-based model that encodes the original function code into a joint graph structure containing rich program semantics. The gated graph recurrent layer learns node features by aggregating and propagating information in the graph, and the Conv module extracts meaningful node representations for graph-level predictions. AMPLE\cite{Wen_Chen_Gao_Zhang_M.Zhang_Liao_2023} simplified node representations based on Devign and designed an edge-aware graph convolutional network module, addressing the challenge of graph neural networks capturing relationships between distant graph nodes. VulCNN\cite{wu2022vulcnn} is a scalable graph-based vulnerability detection system that effectively converts function source code into images while preserving program semantics, suitable for large-scale vulnerability scanning.

\section{CONCLUSION}

In this paper, we propose and validate that high-quality images with strong continuity can effectively improve the detection accuracy of Convolutional Neural Networks (CNNs). Additionally, we introduce the use of CFG-based pixel row oversampling method, which synthesizes new samples by concatenating between CFG node code, enhancing the continuity of code images. We conduct theoretical analysis and parameter studies, design and compare multiple oversampling methods, as well as some advanced vulnerability detection systems. VulMCI achieves high detection performance and low time overhead in both the SARD and NVD datasets, demonstrating its effectiveness.

\section*{Acknowledgment}

This work was supported by the National Natural Science Foundation of China under Grant 62203337. The authors would like to express their gratitude to the creators of the SySeVR dataset and the open-source code provided by the VulCNN authors.

\bibliographystyle{IEEEtran}
\bibliography{references}

\end{document}